\documentclass[american,aps,prl,reprint,tightenlines,superscriptaddress,twocolumn,twoside,longbibliography,nofootinbib]{revtex4-1}
\pdfoutput=1

\usepackage[utf8]{inputenc}
\usepackage{float}
\usepackage[unicode=true,pdfusetitle, bookmarks=true,bookmarksnumbered=false,bookmarksopen=false, breaklinks=false,pdfborder={0 0 0},backref=false,colorlinks=false] {hyperref}
\hypersetup{ colorlinks,linkcolor=myurlcolor,citecolor=myurlcolor,urlcolor=myurlcolor}
\usepackage{braket,colortbl,amsthm,amsmath,amssymb,dsfont}
\usepackage{accents}
\renewcommand*{\dot}[1]{%
  \accentset{\mbox{\bfseries .}}{#1}}
\definecolor{myurlcolor}{rgb}{0,0,0.7}
\usepackage{graphics,graphicx,relsize}
\usepackage{xcolor}
\usepackage{XCharter}
\usepackage[T1]{fontenc}
\usepackage{mathptmx}
\usepackage[left=1.6cm,right=1.6cm,top=2.45cm,bottom=2.45cm]{geometry}
\usepackage[caption=false]{subfig}
\usepackage{bigints}
\usepackage{accents}

\usepackage{enumitem}

\theoremstyle{plain}
\newtheorem{theorem}{Theorem}

\newtheorem{proposition}{Proposition}
\def\bea{\begin{eqnarray}}
\def\eea{\end{eqnarray}}
\def\ba{\begin{array}}
\def\ea{\end{array}}

\def\beq{\begin{equation}}
\def\eeq{\end{equation}}

\usepackage[normalem]{ulem}
\makeatletter
\let\oldsection\section
\renewcommand{\section}[1]{%
  \par\vspace{1ex} 
  \noindent\textit{#1}.--- 
  \ignorespaces
}

\makeatother

\begin{document}
\title{Extracting Work from Discrete Quantum Polytropic Processes}

\author{Vishal Anand}
\email{anandreuz96@gmail.com}
\affiliation{Center for Quantum Science and Technology, International Institute of Information Technology, Hyderabad 500 032, India}
\affiliation{Center for Security, Theory and Algorithmic Research, International Institute of Information Technology, Hyderabad 500 032, India}

\author{Swarup Kumar Giri}
\email{physicsme1729@gmail.com}
\affiliation{School of Physical Sciences, National Institute of Science Education and Research, HBNI, Jatni, Bhubaneswar 752 050, India}

\author{Avijit Misra}
\email{avijitmisra@iitism.ac.in}
\affiliation{Department of Physics, Indian Institute of Technology (ISM), Dhanbad, Jharkhand 826 004, India}

\author{Subhadip Mitra}
\email{subhadip.mitra@iiit.ac.in}
\affiliation{Center for Computational Natural Sciences and Bioinformatics, International Institute of Information Technology, Hyderabad 500 032, India}
\affiliation{Center for Quantum Science and Technology, International Institute of Information Technology, Hyderabad 500 032, India}

\author{Samyadeb Bhattacharya}
\email{samyadeb.b@iiit.ac.in}
\affiliation{Center for Quantum Science and Technology, International Institute of Information Technology, Hyderabad 500 032, India}
\affiliation{Center for Security, Theory and Algorithmic Research, International Institute of Information Technology, Hyderabad 500 032, India}

\begin{abstract}
\noindent
We establish an upper bound on extractable work for time-dependent, non-Markovian quantum heat engines operating with finite baths. This bound analytically isolates the distinct thermodynamic penalties arising from system-bath correlations, bath non-equilibrium, and residual interaction energy. Evaluating this framework operationally via a quantum polytropic cavity-optomechanical cycle, we demonstrate that maximal efficiency requires quasi-static operation to successfully harvest coherent, non-Markovian system-bath resonances. Conversely, optimising for maximum power enforces a strict finite-time regime. Under realistic hardware constraints, this acceleration necessitates larger discrete operational steps, where we expect Trotterisation errors to manifest as physical noise. Such noise would irreversibly suppress delicate quantum memory effects, forcing a collapse to the memoryless Markovian Otto limit. Coupled with the permanent energetic tax of switching finite-bath interactions, our results indicate that the exploitation of quantum memory resources and finite-power operation belong to different operational regimes.
\end{abstract}

\maketitle
\section{Introduction}
Environmental interactions shape the thermodynamic behaviour of most quantum systems~\cite{jaramillo2016,guo2018,ghosh2022,singh2020,TD1, mitchison2019review}. Thus, understanding the fundamental limits of energy conversion remains a central objective of quantum thermodynamics~\cite{TD2,TD4, Vinjanampathy2016}. Thermodynamics quantifies energy exchange primarily through heat and work: heat involves thermal transfer and entropy change, while work represents ordered energy transfer driven by external control. In the quantum regime, this distinction blurs for finite systems \cite{Reeb} undergoing coherent dynamics. Recent experimental advances in trapped ions~\cite{Abah2012, Maslennikov2019}, superconducting circuits~\cite{bohr2022}, quantum dots~\cite{Bergenfeldt2014, levy2017}, and cavity optomechanics~\cite{chen2024} enable the realisation of microscopic thermal machines operating beyond the assumptions of conventional thermodynamics~\cite{pradhan2025}. In these settings, the environment is rarely infinitely large or unaffected by the system. As a result, finite-bath effects, memory-induced dynamics, and system-environment correlations can no longer be ignored. While maximum extractable work is traditionally bounded by the non-equilibrium free-energy difference of the working medium, this relies on large equilibrating baths, negligible interaction energies, and Markovian dynamics. Because realistic devices upend this classical energetic balance, characterising work extraction under finite, non-Markovian conditions requires better thermodynamic accounting in a more general framework.

In this Letter, we establish the thermodynamic limit of work extraction for time-dependent, non-Markovian quantum processes \cite{nm1,nm2,nm3,nm4,nm5,nm6,nm7, Bylicka2016}. To formulate the problem operationally, we introduce a quantum polytropic process constructed from alternating discrete adiabatic and isochoric transformations via a Suzuki-Trotter decomposition. This establishes a natural quantum analogue to classical polytropic processes applicable in non-equilibrating, finite-bath scenarios. Our central result is a general upper bound that analytically separates the maximum extractable work into four distinct physical contributions: the drop in the system's effective free energy, the thermodynamic penalty of creating system-bath correlations, the cost of driving the finite bath out of equilibrium, and the change in interaction energy. While prior finite-bath frameworks~\cite{work1} rely heavily on energy-conserving operations, our bound explicitly accommodates continuous, time-dependent external driving -- extending standard free-energy relations while recovering known limits~\cite{popescu1}. 

We validate this theoretical bound in a physical setup: a cavity platform containing collectively trapped atoms coupled to vibrating mirrors, in which the phononic modes serve as finite thermal environments. Within this architecture, we construct a hybrid quantum thermal cycle that interpolates between Otto and Stirling operations. Our results establish the thermodynamic constraint on work extraction in non-Markovian quantum devices, providing a framework for finite-bath heat engines operating beyond the equilibrium paradigm.

\begin{figure}[t!]
    \centering
    \includegraphics[width=\columnwidth]{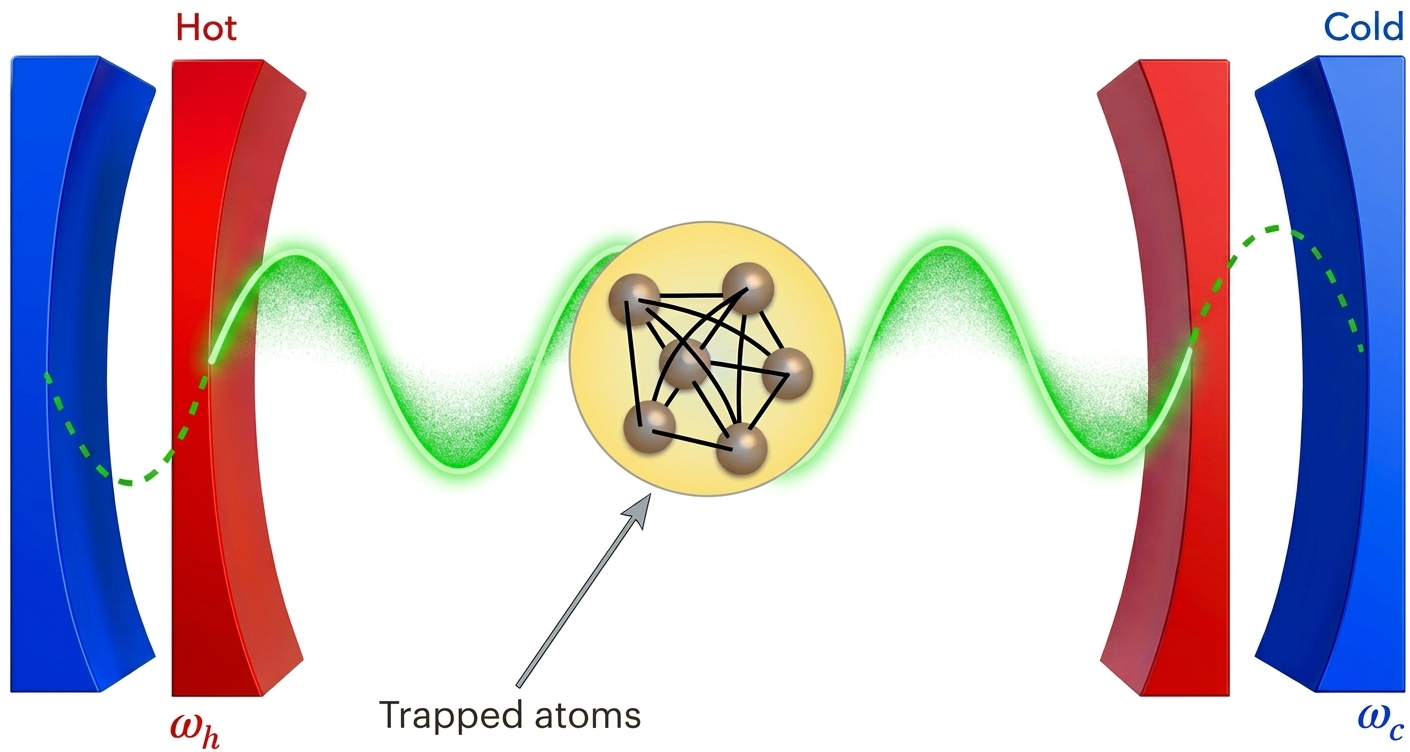}
    \caption{{\bf Physical realisation of a quantum polytropic process:} Uniformly interacting atoms trapped inside a single-mode cavity of adjustable length. The cavity operates between two lengths: at its shortest, the mode frequency is $\omega_h$, coupling the system to the hot bath (red); at its longest, the frequency is $\omega_c$, coupling it to the cold bath (blue).}
    \label{fig:trap}
\end{figure}
\section{Constructing a polytropic process} 
In classical thermodynamics, a polytropic process for an ideal gas generalises isothermal and adiabatic processes. It follows $PV^{\zeta}=$ constant, where $P$ is pressure, $V$ is volume, and the polytropic index $\zeta$ is a real number. For an isothermal process, $\zeta=1$, while for an adiabatic process, $\zeta$ is the ratio of specific heats. By tuning $\zeta$, this equation captures various thermodynamic transformations. To extend this concept to the quantum regime, we first define quantum isochoric and adiabatic processes using the schematic in Fig.~\ref{fig:trap}.

This schematic shows a possible physical setup: a collection of atoms interacting uniformly, behaving as a single quantum entity, trapped in an optical cavity bounded by vibrating mirrors. The cavity length is adjustable, allowing the setup to effectively act as a quantum particle in a box with movable walls. The atomic system exchanges energy with the vibrational modes of the mirrors. Since the system's characteristic frequency is inversely proportional to the cavity length, tuning this length shifts the system into resonance with different mirror frequencies. At its longest, the cavity supports a lower frequency $\omega_c$ coupled to a cold bath at temperature $T_c$, and at its shortest, a higher frequency $\omega_h$ coupled to a hot bath at temperature $T_h$. Using this architecture, we define two fundamental processes in general.
\begin{itemize}[leftmargin=*]
    \item[--]
    {\it Quantum isochoric process:}
    A quantum isochoric process couples the working medium to a thermal reservoir, allowing heat exchange without any work. In our setup, this occurs when the cavity length remains fixed. The system resonates with a specific cavity mode ($\omega_c$ or $\omega_h$), driving thermalisation. We represent the working medium and the mirror vibrations by density matrices $\rho_S$ and $\tau_E$, respectively, where $\rho_S$ is an arbitrary state and $\tau_E$ is the thermal state of the active mode. The joint state evolves via the energy-conserving global unitary $U_{SE}=\exp(-iHt/\hbar)$. Here, $t$ is the interaction time and $H=H_S+H_E+H_{\rm int}$ is the global Hamiltonian. Imposing the condition $[H_{\rm int},~H_S+H_E]=0$ ensures the global operation remains energy preserving. The resulting isochoric map on the working medium is therefore $\Phi_{\rm iso}(\rho_S)=\text{Tr}_E\left[U_{SE}(\rho_S\otimes\tau_E) U_{SE}^\dagger\right]$.
    
    \item[--]
    {\it Quantum adiabatic process:}
    Like its classical counterpart, a quantum adiabatic process involves no heat exchange. While traditionally associated with invariant probability distributions~\citep{nori1, Zhang2014}, extracting work inevitably alters state occupation probabilities. We therefore model the quantum adiabatic process broadly as a local, time-dependent unitary evolution of the working medium. In our setup, this corresponds to expanding or contracting the cavity length. This unitary driving shifts the system's characteristic frequency between $\omega_c$ and $\omega_h$, directly manipulating the energy-level spacings. The adiabatic map on the working medium is thus
    $\Phi_{\rm ad}(\rho_S)=U_S\rho_S U_S^\dagger$.
\end{itemize}
Below, we map these general fundamental thermodynamic operations to a specific microscopic model for a case study.

\section{Quantum polytropic process}
In an ideal isothermal process, the cavity length would change infinitely slowly while the system thermalises. However, strict isothermal conditions are unrealistic here because the temperature of the finite vibrational mode can fluctuate during the interaction. To resolve this, we construct a general quantum polytropic process by alternating infinitesimal adiabatic and isochoric steps~\citep{Anders2013, Baumer2019, Chen2021, yuan2024}. Each adiabatic step marginally changes the cavity length, while each isochoric step allows partial thermalisation.

Let $\mathcal{U}_S^{(k)}(\epsilon_1)=\exp(-iH_{\rm adi}^{(k)}\epsilon_1/\hbar)$ represent the local unitary driving on the system for duration $\epsilon_1$, and $\mathcal{U}_{SE}^{(k)}(\epsilon_2)=\exp(-iH_{\rm iso}^{(k)}\epsilon_2/\hbar)$ represent the global system-mode interaction for duration $\epsilon_2$ during the $k$-th iteration (here, the mode essentially acts as the environment). The full quantum polytropic map is then
\begin{align}
    \Phi_{\rm pol}(\rho_S)= \text{Tr}_E&\left[\left(\prod_{j=\mathcal{M}}^{1}\mathcal{U}^{(j)}_{SE}(\epsilon_2)\mathcal{U}^{(j)}_S(\epsilon_1)\right)\left(\rho_S\otimes\tau_E\right)\right.\nonumber\\
    &\ \left.\left(\prod_{k=1}^{\mathcal{M}}\mathcal{U}^{(k)\dagger}_S(\epsilon_1)\mathcal{U}^{(k)\dagger}_{SE}(\epsilon_2)\right)\right].
\end{align}
Here, the driving shifts the system's characteristic frequency from an initial value $\omega_{\rm ini}$ to a final value $\omega_{\rm fin}$. The total duration of each iteration is $\epsilon = \epsilon_1 + \epsilon_2$, with the time division parameterised by $\kappa \in [0,1]$ such that $\epsilon_1 = \kappa\epsilon$ and $\epsilon_2 = (1-\kappa)\epsilon$. The parameter $\kappa$ controls the ratio of coherent driving to thermalisation, serving as the operational quantum analogue to the classical polytropic index $\zeta$. Unlike standard repeated-interaction collision models that enforce Markovianity by continually discarding environmental ancillae, our global unitary framework retains the persistent finite bath, explicitly permitting the buildup of system-environment correlations and non-Markovian memory. The limit $\kappa \to 1$ implies purely adiabatic unitary evolution, $\kappa \to 0$ is isochoric thermalisation, and intermediate values interpolate between them by dictating the effective heat capacity of the stroke. By taking a large number of steps $\mathcal{M}$ over a total operation time $t=\mathcal{M}\epsilon$, we can evaluate the reduced dynamics of the polytropic process $\Phi_{\rm pol}$. Crucially, while each isochoric step conserves energy locally, the total global unitary $U_{\rm pol}=\prod_{k=\mathcal{M}}^1\mathcal{U}^{(k)}_{SE}(\epsilon_2)\mathcal{U}^{(k)}_S(\epsilon_1)$ is driven by a time-dependent Hamiltonian and does not conserve energy. This provides a physically realisable thermodynamic process from which we can derive the fundamental limits on extractable work.

\section{Extractable work}
For any non-adiabatic polytropic stroke, the working medium is subjected to environmental interaction. The system and environment evolve under a global unitary operation $\rho_S \otimes \tau_E \rightarrow U_g (\rho_S \otimes \tau_E) U_g^\dagger$, generated by the total time-dependent Hamiltonian $H(t)=H_S(t)+H_E+H_{\rm int}(t)$. The rate of change in total energy is thus $\dot{E}_{\rm total} = \dot{E}_S+\dot{E}_E+\dot{E}_{\rm int}$, where $E_r=\text{Tr}[H_r U_g(\rho_S\otimes\tau_E) U_g^\dagger]$ for $r \in \{S,E,\text{int}\}$. Integrated over the process, global energy conservation dictates that the total energy drop bounds the maximum work transferred to an external battery, $W_{\rm ext} \le -\Delta E_{\rm total}$~\cite{Alipour2016}. We identify $Q=-\Delta E_E$ as the heat exchanged with the finite bath. Moreover, because the residual interaction energy is inaccessible for useful extraction, we define the generated work on the working medium as $W_{\rm gen}=\Delta E_{\rm total}-\Delta E_{\rm int}$. This yields a first-law-like relation for our setup: $\Delta E_S = Q+W_{\rm gen}$.

While $W_{\rm gen}$ is closely associated with the available free energy, the finiteness of the bath and the buildup of correlations prevent it from being fully extractable. Extending the formalism of Ref.~\citep{work1}, we establish the fundamental limit on useful work for this architecture.

\begin{theorem}[Upper bound on extractable work]\label{thm:1}
The extracted work $W_{\rm ext}$ from a bipartite system consisting of a working medium $S$ and a finite bath $E$ undergoing a global unitary evolution with time-dependent external driving, and operating within a continuous non-Markovian cycle, is upper-bounded by:
\begin{align}\label{eq:thmbound}
W_{\rm ext} \leqslant -\Delta F_S^{\rm eff} - W_{\rm drive} - \frac{1}{\beta} \left[ \Delta I_{S:E} + \Delta S_E^{\rm eff} \right] - \Delta E_{\rm int},
\end{align}
where $-\Delta F_S^{\rm eff}$ is the drop in the effective non-equilibrium free energy of the system, $W_{\rm drive} = -\ln(Z^\prime/Z)/\beta$ is the work done by the external drive on the system, $\Delta I_{S:E} = I_{S:E}^\prime - I_{S:E}$ is the change in system-bath mutual information, $\Delta S_E^{\rm eff}$ quantifies the change in the bath's relative entropy with respect to its effective thermal state, and $\Delta E_{\rm int}$ is the change in the interaction energy.   
\end{theorem}

We present a short, intuitive sketch of the proof; the detailed proof of the theorem is found in the Supplementary Material. First, we model work extraction by introducing an explicit, weakly coupled battery. Global energy conservation dictates that the maximum work transferred to this battery is bounded by the total energy drop of the system-environment composite. We then evaluate this bipartite energy drop using the fundamental identity connecting energy changes to the evolution of the joint relative entropy $S(\rho_{SE}||\tau_{SE})$ with respect to the instantaneous thermal state. Finally, decomposing this global relative entropy into local marginals, system-bath mutual information, and the bath's deviation from equilibrium exactly isolates the physical thermodynamic penalties expressed in the theorem. 

For the special case of large and static baths weakly interacting with the system, $\Delta I_{S:E}$ and $\Delta S_E^{\rm eff}$ become negligible. Additionally, if there is no external periodic driving present in the dynamics, the upper bound reduces to $\widetilde W_{\rm upp}=-\Delta E_S+\frac{1}{\beta}\Delta S_S-\Delta E_{\rm int}$. For a process governed by an energy-preserving global unitary where the interaction energy is a constant of motion ($\Delta E_{\rm int} = 0$), this bound then recovers the known limit: the change in the free energy of the working medium, $\Delta F = \Delta E_S - \Delta S_S/\beta$~\cite{popescu1}.

\begin{figure*}
    \captionsetup[subfigure]{labelformat=empty}
    \centering
    \subfloat{\includegraphics[width=0.49\textwidth]{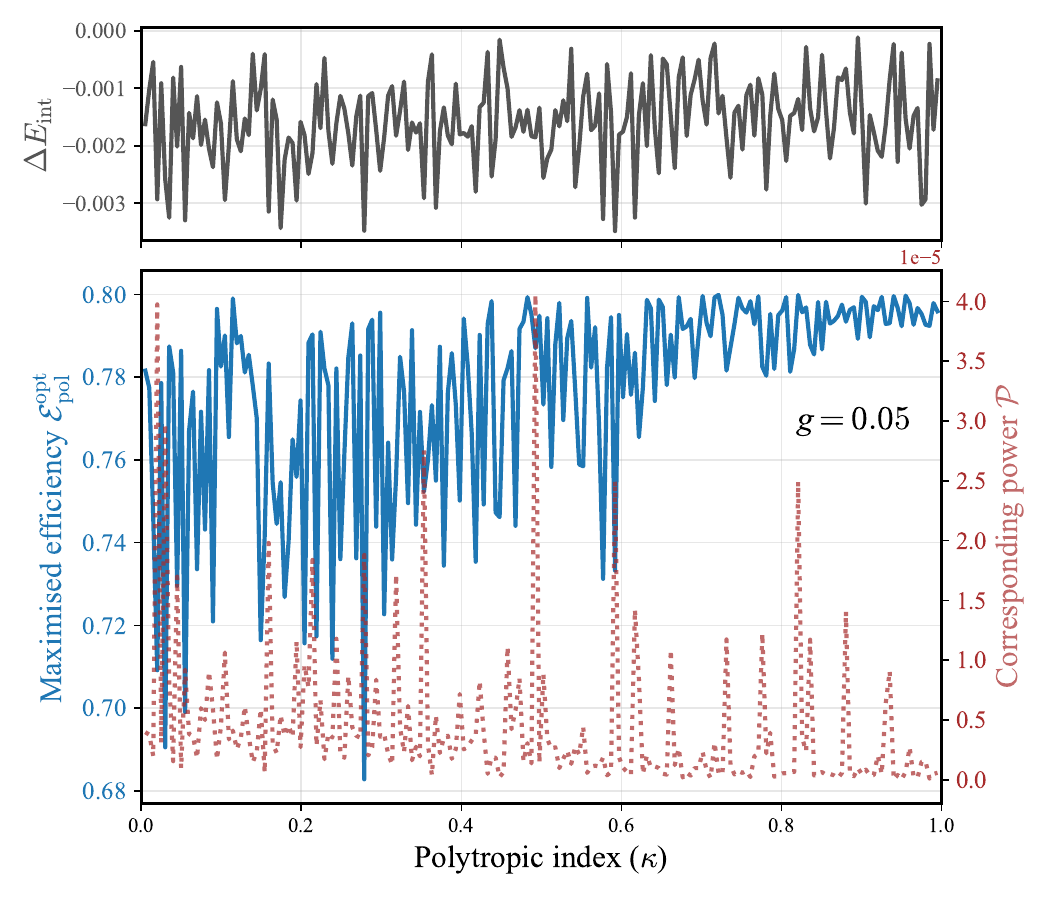}\label{fig:max_eff}}\hfill
    \subfloat{\includegraphics[width=0.49\textwidth]{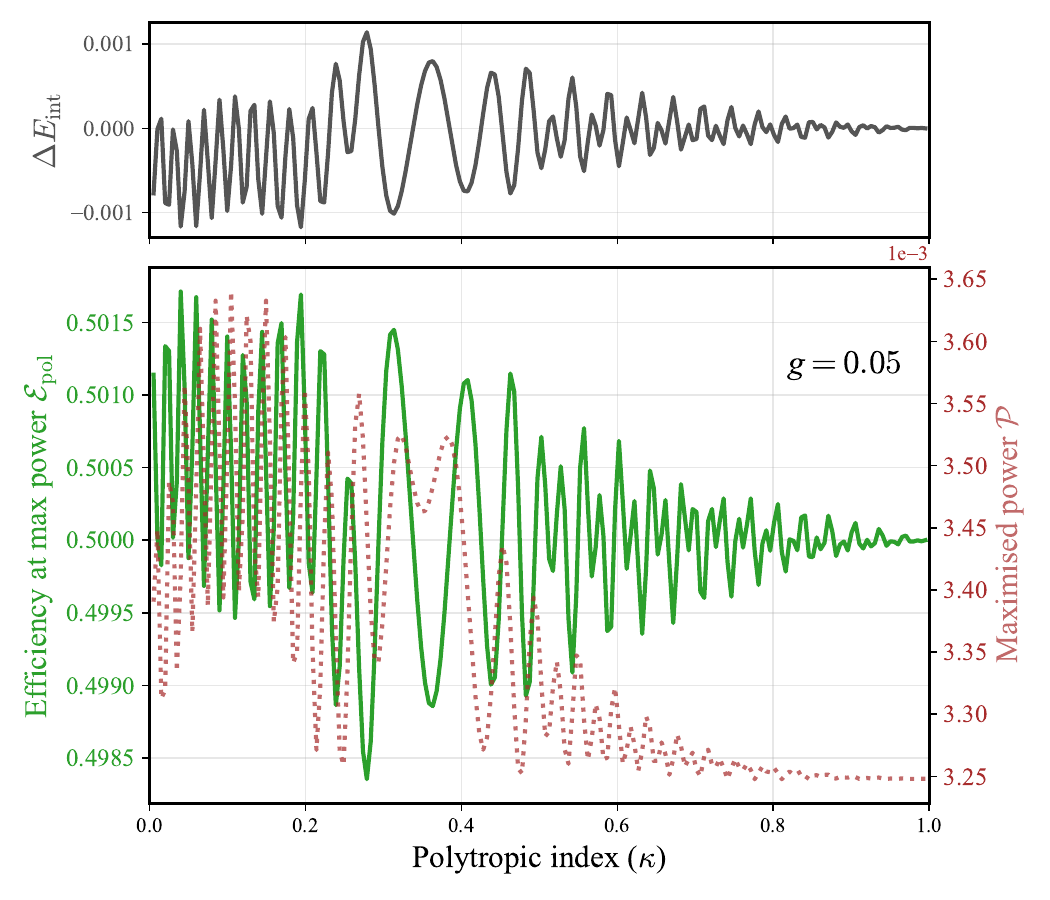}\label{fig:max_pow}}
    \caption{\textbf{Thermodynamic performance and the finite-bath interaction cost of the hybrid cavity-optomechanical cycle.} The main panels plot performance bounds as a function of the polytropic index $\kappa$. To accurately capture the increasingly sparse quasi-static resonances as $\kappa \to 1$, the optimisation is performed over an adaptive interaction time grid $t \in [100, t_{\max}(\kappa)]$, scaling up to $t_{\max} \approx 22000$ (where $t = \mathcal{M}\epsilon$). The top narrow panels display the corresponding residual interaction cost $\Delta E_{\text{int}}$ incurred per cycle. {\bf (Left)} When optimised for maximum efficiency (blue), the engine exploits non-Markovian memory effects, producing sharp oscillatory resonances but vanishing power (red dotted line). This quasi-static regime incurs an energy cost (top panel) to sever the deep system-bath correlations. {\bf (Right)} Optimising for maximum power (red dotted line) enforces a finite-time regime where Trotterisation errors produce noise, irreversibly suppressing the fragile coherent resonances and forcing an efficiency collapse (green line) to the Markovian limit. However, fast driving physically prevents any system-bath correlation build-up, dropping the decoupling cost to zero. All thermodynamic quantities are evaluated after the working medium converges to a stable limit cycle. We set $\hbar=k_B=1$. System parameters: $\omega_h = 2.0$, $\omega_c = 1.0$, $\beta_h = 0.2$, $\beta_c = 1.0$, and coupling strength $g = 0.05$. (See Fig.~\ref{fig:performancevariance} in the Supplementary Material for more choices of $g$.)  \label{fig:performance}} 
\end{figure*}
\section{Microscopic model and hybrid thermal cycle} 
To evaluate our bound operationally, we construct a quantum thermodynamic cycle based on the schematic depicted in Fig.~\ref{fig:trap}. First, we restrict the working medium to a two-level system. For a system comprising $N$ identical two-level atoms, this is justified by restricting the energetic transitions. Let us assume there are only two possibilities: either all the atoms are in their ground state $\ket{0_1 0_2 \dots 0_N}$, or only one of them is in its excited state $\ket{0_1 0_2 \dots 1_i \dots 0_N}$. The restricted two-level subspace thus has a collective ground state $\ket{0}=\ket{0_1 0_2 \dots 0_N}$ and an excited state defined as an equal superposition of all possible single excitations, $\ket{1}=\frac{1}{\sqrt{N}}\sum_{i=1}^N\ket{0_1 0_2 \dots 1_i \dots 0_N}$. The Hamiltonian of the working medium can be written as $H_S=\hbar\omega_0\sigma_z$, with $\sigma_z$ being the diagonal Pauli matrix. 

The mirror vibrations, acting as the environment, can be modelled as an infinite collection of oscillators with the Hamiltonian $H_{E_{\rm tot}}=\sum_{i=1}^\infty H_{E_i}=\sum_{i=1}^\infty\hbar\omega_i a_i^\dagger a_i$. We further restrict ourselves to the case where the working medium energetically interacts with only the ground and first excited states of the oscillators. Note that at any given time, the working medium interacts with only a single active mode. We assume there are two relevant collective mirror vibrations with modes in thermal states identified by two inverse temperatures, $\beta_c$ and $\beta_h$ ($\beta_c > \beta_h$). All modes are part of one of these two groups, and if left in isolation from the working medium, they quickly thermalise back to equilibrium. Under resonant conditions ($\omega_0 = \omega_i$), the Hamiltonian of the proposed schematic can therefore be written as 
\begin{equation}\label{eq:ham1}
    H=\hbar\omega_i\sigma_z\otimes\mathbb{I}+\hbar\omega_i\mathbb{I}\otimes\sigma_z+\hbar g\left(\sigma_-\otimes\sigma_++\sigma_+\otimes\sigma_-\right).
\end{equation}

With this architecture in mind, we design a hybrid thermal cycle of alternating isochoric and polytropic strokes. When the polytropic processes reduce to purely adiabatic or isothermal limits, the cycle recovers the quantum Otto and Stirling engines, respectively. By tuning $\kappa$, we can balance the efficiency of a Stirling engine against the speed of an Otto engine, minimising the thermodynamic costs of using finite baths. The first stage ($A\rightarrow B$) is a polytropic expansion. The initial system frequency $\omega_h$, which is in resonance with the hot bath, is slowly decreased step-by-step to $\omega_c$ by expanding the cavity length, bringing it into resonance with the cold bath. The next step ($B\rightarrow C$) is isochoric cooling; the working medium remains in resonance with the cold bath, dumping heat via thermalisation. This is followed by a polytropic compression ($C\rightarrow D$), where the characteristic frequency is slowly increased from $\omega_c$ back to $\omega_h$ by contracting the cavity length. Finally ($D\rightarrow A$), an isochoric heating stroke extracts heat from the hot bath to complete the cycle.

To estimate its efficiency, we determine the work extracted during the polytropic contraction and expansion ($W_{\rm con}$ and $W_{\rm exp}$, respectively). Heat is absorbed from the hot bath during the isochoric heating ($Q_{\rm iso}^h$) and the polytropic expansion ($Q_{\rm pol}^h$). From Eq.~\eqref{eq:thmbound}, we see that the actual efficiency is capped by the upper bound on extractable work, yielding an optimal efficiency limit:
\begin{equation}
 \mathcal{E}_{\rm pol}\leqslant\mathcal{E}_{\rm pol}^{\rm opt}=\frac{W^{\rm upp}_{\rm con}+W^{\rm upp}_{\rm exp}}{Q_{\rm iso}^h+Q_{\rm pol}^h},   
\end{equation}
where $W^{\rm upp}_{\rm con}$ and $W^{\rm upp}_{\rm exp}$ are the theoretical upper bounds for the extracted work during the respective polytropic strokes. 

For the device to function operationally as a heat engine, the initial state populations must satisfy the thermodynamic condition: $q_{h,0} < p_A < (q_{c,0}-a q_{h,0})/(1-a)$, where $0\leqslant a\leqslant 1$, and $p_A$, $q_{h,0}$, and $q_{c,0}$ are the ground state populations of the initial system state, the hot bath, and the cold bath, respectively. Physically, this inequality ensures the working medium is sufficiently colder than the hot bath to absorb heat during expansion, yet remains hot enough to dump heat into the cold bath during compression. Violating these population bounds reverses the macroscopic heat flows, transitioning the cycle into a refrigerator or a heater.

\section{Discussion} 
Within this operational regime, the tension between theoretical thermodynamic ceilings and practical engine performance becomes particularly evident as the cycle approaches the adiabatic limit ($\kappa \to 1$, see Fig.~\ref{fig:performance}). When optimised for maximum efficiency, the engine compensates for the vanishing system-bath coupling by pushing the interaction time toward the quasi-static limit. The extended timescale allows the working medium to exploit fragile non-Markovian memory resonances, achieving near-Carnot efficiencies at the cost of vanishing power output. Conversely, when the cycle is driven to maximise power, the engine is forced to operate in finite time. This finite-time constraint precludes the system from resolving these slow, delicate resonances; the coherent system-bath ringing is entirely suppressed. Consequently, the dynamics collapse into a standard, decoupled quantum Otto cycle, where the efficiency at maximum power approaches the limit, $\eta_{\text{Otto}} = 1 - \omega_c/\omega_h = 0.5$.

Underpinning this operational trade-off is the fundamental asymmetry of the thermodynamic penalties derived in our upper bound. The bound restricts extractable work via both system-bath correlations ($\Delta I_{S:E}$) and bath non-equilibrium ($\Delta S_{E}^{\rm eff}$). Yet these terms originate from distinct physical mechanisms. The mutual information represents theoretically recoverable locked work, whereas the finite bath's divergence from its equilibrium state reflects irreversible dissipation. The working medium can only mitigate these penalties by exploiting non-Markovian memory in the quasi-static limit. When the cycle is accelerated to maximise power, the engine is forced to operate in a strict finite-time regime (small $t$). Because physical hardware requires a minimum feasible switching time ($\epsilon$), operating at high speeds reduces the number of discrete steps ($\mathcal{M}$). Under this hardware constraint, we expect the formally neglected $\mathcal{O}(\epsilon^2)$ Trotter errors to manifest as physical dephasing noise. Such noise, coupled with the finite-time constraint, would wash out the delicate coherent resonances and force a collapse to the Markovian limit. The change in interaction energy $\Delta E_{\rm int}$ does not vanish over a complete periodic cycle; instead, operating in the quasi-static regime to maximise efficiency incurs a substantial, oscillatory energetic cost to sever the deep non-Markovian correlations. Conversely, operating in the finite-time limit physically prevents this system-bath correlation build-up, effectively dropping the decoupling cost to zero. Thus, while Trotter dephasing noise is expected to drive the efficiency collapse to the Markovian limit, our numerics confirm that fast driving simultaneously prevents the accumulation of the finite-bath decoupling penalty. Ultimately, under realistic experimental constraints, our results indicate that the exploitation of non-Markovian resources and finite-power operation belong to different operational regimes for systems operating in discrete steps.

\section{Acknowledgement}
 A.M. acknowledges support
from the Anusandhan National Research Foundation through PMECRG grant (Grant No. ANRF/ECRG/2024/003836/PMS). SM and SB acknowledge support from the Ministry of Electronics and Information Technology (MeitY), Government of India, under Grant No. 4(3)/2024-ITEA.
\bibliography{reference}

\newpage
\onecolumngrid
\let\section\oldsection
\renewcommand{\theequation}{S\arabic{equation}}
\setcounter{theorem}{0}
\setcounter{equation}{0}
\renewcommand{\thefigure}{S\arabic{figure}}
\renewcommand{\theHfigure}{S.\arabic{figure}}
\setcounter{figure}{0}
\section*{Supplementary Material}
\section*{Proof of the theorem}\label{sec:thm}
\noindent
\begin{theorem}[Upper bound on extractable work]
The extracted work $W_{\rm ext}$ from a bipartite system consisting of a working medium $S$ and a finite bath $E$ undergoing a global unitary evolution with time-dependent external driving, and operating within a continuous non-Markovian cycle, is upper-bounded by:
\begin{align*}
W_{\rm ext} \leqslant -\Delta F_S^{\rm eff} - W_{\rm drive} - \frac{1}{\beta} \left[ \Delta I_{S:E} + \Delta S_E^{\rm eff} \right] - \Delta E_{\rm int},
\end{align*}
where $-\Delta F_S^{\rm eff}$ is the drop in the effective non-equilibrium free energy of the system, $W_{\rm drive} = -\ln(Z^\prime/Z)/\beta$ is the work done by the external drive on the system, $\Delta I_{S:E} = I_{S:E}^\prime - I_{S:E}$ is the change in system-bath mutual information, $\Delta S_E^{\rm eff}$ quantifies the change in the bath's relative entropy with respect to its effective thermal state, and $\Delta E_{\rm int}$ is the change in the interaction energy.   
\end{theorem}
\begin{proof}
Let us consider the working medium and bath evolving under a global unitary operation from an initial joint state $\rho_{SE}$ to a final state $\rho_{SE}^\prime$. The dynamics are generated by a time-dependent total Hamiltonian $H(t) = H_S(t) + H_E + H_{\rm int}(t)$. The time-dependence in $H_S(t)$ explicitly accounts for the unitary driving of the system, such as the modulation of the cavity length during the adiabatic and polytropic strokes. To realise the extractable work in the described setup, we consider a work-storing battery (weight) as an additional quantum system $\rho_W$, such that the now global operation is represented by an energy-preserving unitary $U_{SEW}=U_{SE} U_{SW}$. Here, we assume that the battery is weakly interacting with only the working medium. Due to the weak interaction between the system and the weight, the globally evolved state $\rho_{SEW}=U_{SEW}(\rho_{S}\otimes\tau_{E}\otimes\rho_W) U_{SEW}^\dagger$ can be assumed as $\rho_{SEW}\approx\rho_{SE}^\prime\otimes\rho_W^\prime$ with $\rho_W^\prime=\text{Tr}_{SE}[\rho_{SEW}]$. The weak system-battery coupling also allows us to assume that $\rho_{SE}^\prime\approx U_{gt}(\rho_{S}\otimes\tau_E) U_{gt}^\dagger$, where $U_{gt}$ is now a non-energy-conserving unitary with a time-dependent external driving. Hence, the extracted work stored in the battery can now be considered as $W_{\rm ext}\leqslant \text{Tr}[H_W(\rho_W^\prime-\rho_W)]$. The inequality comes because some energy can be locked in correlation. Therefore, global energy conservation gives us $W_{\rm ext}\leqslant -(\Delta E_S+\Delta E_E+\Delta E_{\rm int})$. Let us define the bare system-bath Hamiltonian as $H_0(t) = H_S(t) + H_E$. We can express the extractable work inequality in terms of the bare energy drop:
\begin{align*}
W_{\rm ext} \leqslant \text{Tr}[H_0\rho_{SE}]-\text{Tr}[H_0^\prime\rho_{SE}^\prime] - \Delta E_{\rm int}.
\end{align*}

We now consider the identity for the change in relative entropy between the actual joint state and the instantaneous thermal state $\tau_{SE} ~(\tau^\prime_{SE}) = e^{-\beta H_{S}}/Z_S\otimes e^{-\beta H_{E}}/Z_E \ (e^{-\beta H^\prime_{S}}/Z^\prime_S\otimes e^{-\beta H^\prime_{E}}/Z^\prime_E)$. Because the thermal state is defined with respect to the bare Hamiltonian $H_0$, the identity is:
\begin{align*}
S(\rho_{SE}||\tau_{SE}) -& S(\rho_{SE}^\prime||\tau_{SE}^\prime) 
= \ \beta(\text{Tr}[H_0\rho_{SE}] - \text{Tr}[H_0^\prime\rho_{SE}^\prime]) - \ln\left(\frac{Z^\prime}{Z}\right).
\end{align*}
Here $Z~(Z^\prime)=Z_SZ_E~(Z^\prime_SZ^\prime_E)$. 
By rearranging our energy conservation bound, we know that $\text{Tr}[H_0\rho_{SE}]-\text{Tr}[H_0^\prime\rho_{SE}^\prime] \geqslant W_{\rm ext} + \Delta E_{\rm int}$. Substituting this into the relative entropy identity gives:
\begin{align*}
S(\rho_{SE}||\tau_{SE}) - S(\rho_{SE}^\prime||\tau_{SE}^\prime) \geqslant \beta(W_{\rm ext} + \Delta E_{\rm int}) - \ln\left(\frac{Z^\prime}{Z}\right).
\end{align*}
Rearranging to isolate the extracted work:
\begin{align*}
W_{\rm ext} \leqslant \frac{1}{\beta} \ln\left(\frac{Z^\prime}{Z}\right) + \frac{1}{\beta} \left[ S(\rho_{SE}||\tau_{SE}) - S(\rho_{SE}^\prime||\tau_{SE}^\prime) \right] - \Delta E_{\rm int}.
\end{align*}
We identify the term associated with the changing partition function as the external driving work, $W_{\rm drive} = F_\tau^\prime - F_\tau = - \ln(Z^\prime/Z)/\beta$. Because $H_S(t)$ is time-dependent across the stroke, $Z \neq Z^\prime$, and this driving contribution must be retained. Replacing the partition function ratio yields:
\begin{align*}
W_{\rm ext} \leqslant -W_{\rm drive} + \frac{1}{\beta} \left[ S(\rho_{SE}||\tau_{SE}) - S(\rho_{SE}^\prime||\tau_{SE}^\prime) \right] - \Delta E_{\rm int}.
\end{align*}

Next, to account for continuous operation in a non-Markovian regime, we decompose the global relative entropy for both the initial and final states into their respective marginal contributions, mutual information, and an interaction mismatch parameter:
\begin{align*}
S(\rho_{SE}^{x}||\tau_{SE}^{x}) = S(\rho_S^{x}||\tau_S^{x}) + S(\rho_E^{x}||\tau_E^{x}) + I_{S:E}^{x} ,
\end{align*}
where $x$ indicates whether the state is at the beginning (unprimed) or the end (primed) of the stroke, and $\tau_{S(E)}$ denotes the system (bath) Gibbs state. Substituting these decompositions into our bound gives:
\begin{align*}
W_{\rm ext} \leqslant -W_{\rm drive}&\ + \frac{1}{\beta} \left[ \left( S(\rho_S||\tau_S) - S(\rho_S^\prime||\tau_S^{\prime}) \right) - \Delta S_E^{\rm eff} - \Delta I_{S:E}  \right] - \Delta E_{\rm int},
\end{align*}
where we have defined the changes across the stroke as: $\Delta S_E^{\rm eff} = S(\rho_E^\prime||\tau_E^{\prime}) - S(\rho_E||\tau_E)$ and $\Delta I_{S:E} = I_{S:E}^\prime - I_{S:E}$. Assuming weak coupling and a time-independent bath Hamiltonian, the effective thermal state of the bath remains approximately constant ($\tau_E \approx \tau_E^{\prime}$), allowing this term to be understood purely as the change in the bath's divergence from equilibrium. Finally, recognising that the drop in the system's effective non-equilibrium free energy is defined as $-\Delta F_S^{\rm eff} = \left[ S(\rho_S||\tau_S) - S(\rho_S^\prime||\tau_S^{\prime})\right]/\beta$, we get the exact bound:\begin{align*}
W_{\rm ext} \leqslant -\Delta F_S^{\rm eff} - W_{\rm drive} - \frac{1}{\beta} \left[ \Delta I_{S:E} + \Delta S_E^{\rm eff} \right] - \Delta E_{\rm int},
\end{align*}
which completes the proof.
\end{proof}

\section*{The general polytropic process}
\noindent
The global unitary action of the polytropic process is given by 
\begin{equation}
    \rho_S(0)\otimes\tau_E\rightarrow 
    \left[\left(\prod_{j=\mathcal{M}}^1\mathcal{U}^{(j)}_{SE}(\epsilon_2)\mathcal{U}^{(j)}_S(\epsilon_1)\right)\left(\rho_S(0)\otimes\tau_E\right)\left(\prod_{k=1}^{\mathcal{M}}\mathcal{U}^{(k)\dagger}_S(\epsilon_1)\mathcal{U}^{(k)\dagger}_{SE}(\epsilon_2)\right)\right].
\end{equation}
The Hamiltonians of the $k$-th adiabatic and isochoric steps are given by 
\renewcommand{\arraystretch}{1.5}
\begin{align*}
     H_{\rm adi}^{(k)} &= \left(\omega_{\rm ini}+\frac{k (\omega_{\rm fin}-\omega_{\rm ini})}{\mathcal{M}}\right) \, \sigma_z \otimes \mathbb{I}\quad \text{and}\quad \\
     H_{\rm iso}^{(k)} &= \left(\omega_{\rm ini}+\frac{k (\omega_{\rm fin}-\omega_{\rm ini})}{\mathcal{M}}\right)\sigma_z\otimes\mathbb{I}+\omega_{\rm ini}\mathbb{I}\otimes\sigma_z+ g\left(\sigma_-\otimes\sigma_++\sigma_+\otimes\sigma_-\right),
\end{align*}
respectively, where $g$ is the system-bath coupling strength. This allows us to write the total unitary as
\begin{align}
U_{{\rm pol}} = \prod_{k=\mathcal{M}}^1\mathcal{U}^{(k)}_{SE}(\epsilon_2)\mathcal{U}^{(k)}_S(\epsilon_1)
=\prod_{k=\mathcal{M}}^1 e^{{-i\big( H^{(k)}_{\rm iso}\,\epsilon_2 + H^{(k)}_{\rm adi}\,\epsilon_1 \big)
\;-\;
\frac{\epsilon_2\,\epsilon_1}{2}\,
\big[ H^{(k)}_{\rm iso},\, H^{(k)}_{\rm adi} \big]
}} = \prod_{k=\mathcal{M}}^1e^{A^{(k)}},
\end{align}
where $A^{(k)} = - i \big( H^{(k)}_{\rm iso}\,\epsilon_2 + H^{(k)}_{\rm adi}\,\epsilon_1 \big)
\;-\;
(\epsilon_2\,\epsilon_1/2)\,
\big[ H^{(k)}_{\rm iso},\, H^{(k)}_{\rm adi} \big]$, after ignoring terms higher-order in $\epsilon_1$ and $\epsilon_2$. Since $\big[ H^{(k)}_{\rm iso},\, H^{(k)}_{\rm adi} \big]=- g\left(\omega_{\rm ini}+k (\omega_{\rm fin}-\omega_{\rm ini})/\mathcal{M}\right)\left(\sigma_-\otimes\sigma_+-\sigma_+\otimes\sigma_-\right)$, we get
\begin{align}
    A^{(k)} =& -i \omega_k (\epsilon_1 + \epsilon_2) \sigma_z \otimes \mathbb{I} - i \omega_{\rm ini} \epsilon_2 \mathbb{I} \otimes \sigma_z - i \epsilon_2 g\, X_{{\rm int}} - \epsilon_1 \epsilon_2 g\, \omega_k Y_{{\rm int}},    
\end{align} 
where $\omega_k = \omega_{\rm ini} + k(\omega_{\rm fin}-\omega_{\rm ini})/\mathcal{M}$, $X_{{\rm int}} = \sigma_- \otimes \sigma_+ +\sigma_+ \otimes \sigma_-$, and $Y_{{\rm int}} = \sigma_- \otimes \sigma_+-\sigma_+ \otimes \sigma_-$.
Now, using the identity $e^{A^{k+1}} e^{A^{k}}= e^{A^{k+1} + A^{k} + \tfrac12 [A^{k+1}, A^{k}]}$, we can rewrite $U_{\rm pol}$ as, 
\begin{align}
    U_{\rm pol}=\exp\left(\sum_{k=1}^{\mathcal{M}} A^{(k)}\;+\;\frac{1}{2}\left(\sum_{d=1}^{\mathcal{M}-1} \sum_{k=1}^{\mathcal{M}-d} [A^{k+d}, A^{k}]\right)\right).        
\end{align}
With ${[A^{(k')}, A^{(k)}]} =2 g(k'-k)(\Delta\omega) (\epsilon_2^2 + \epsilon_2\epsilon_1) Y_{\rm int}$, these two terms can be further expressed as
\begin{align}
    \sum_{k=1}^{\mathcal{M}} A^{(k)} =&\ -i(\epsilon_1 + \epsilon_2)\mathcal{M} \left(\omega_{\rm ini} + \frac{(\mathcal{M}-1)}{2}\Delta\omega\right)\sigma_z \otimes \mathbb{I} - i \mathcal{M} \epsilon_2 \omega_{\rm ini} \mathbb{I} \otimes \sigma_z\nonumber\\
    &\ - i g \mathcal{M} \epsilon_2 X_{\text{int}} - \epsilon_1 \epsilon_2  g \mathcal{M} \left(\omega_{\rm ini} + \frac{(\mathcal{M}-1)}{2}\Delta\omega\right) Y_{\text{int}}\\ 
    \text{and} \quad \frac{1}{2} \sum_{d=1}^{\mathcal{M}-1} \sum_{k=1}^{\mathcal{M}-d} [A^{(k+d)}, A^{k}] =&\ \frac{1}{12} \Big( 2\epsilon_2(\epsilon_1 + \epsilon_2) g (\Delta\omega) \mathcal{M}(\mathcal{M}^2 - 1) Y_{\text{int}} \Big),
\end{align}
where $\Delta\omega= \omega_{\rm fin}-\omega_{\rm ini}$ is the change in frequency from the initial $\omega_{\rm ini}$ to the final $\omega_{\rm fin}$. Thus, we get,
\begin{align}
    S=\sum_{k=1}^{\mathcal{M}} A^{(k)} + \frac{1}{2} \left( \sum_{d=1}^{\mathcal{M}-1} \sum_{k=1}^{\mathcal{M}-d} [A^{k+d}, A^{(k)}] \right) =&\ -i(\epsilon_1 + \epsilon_2)\eta \sigma_z \otimes \mathbb{I} - i \mathcal{M} \epsilon_2 \omega_{\rm ini} \mathbb{I} \otimes \sigma_z\nonumber\\
    &\ - i g \mathcal{M} \epsilon_2 X_{\text{int}} + g \Big( -\epsilon_1\epsilon_2\eta + \epsilon_2(\epsilon_1 + \epsilon_2)\mu \Big) Y_{\text{int}},
\end{align}
where $\eta = \mathcal{M} \omega_{\rm ini} + \Delta\omega(\mathcal{M}+1)\mathcal{M}/2 $ and $\mu = \mathcal{M}(\mathcal{M}^2-1)\Delta\omega/6$. Since  $\mathcal{M}\epsilon_2 = \mathcal{M}(1-\kappa)\epsilon = (1-\kappa)t$, we obtain
\begin{equation}
S = -i t \overline{\omega} \sigma_z \otimes \mathbb{I} - i(1-\kappa)t \omega_{\rm ini} \mathbb{I} \otimes \sigma_z - i(1-\kappa)t g X_{{\rm int}} + \frac{1}{6}  g (1-\kappa)t^2 \Delta\omega Y_{{\rm int}},
\end{equation}
where $\overline{\omega} = (\omega_{\rm ini} + \omega_{\rm fin})/2$.

We represent $S$ in the computational basis: $\{|11\rangle, |10\rangle, |01\rangle, |00\rangle\}$. Due to the excitation-preserving nature of the interaction, the matrix naturally decouples into a direct sum of a $1 \times 1$ fully excited state block ($S_{11}$), a $2 \times 2$ single-excitation mixed block ($S_{\rm mid}$), and a $1 \times 1$ fully ground state block ($S_{00}$):
\begin{equation}
    S = S_{11} \oplus S_{\rm mid} \oplus S_{00}.
\end{equation}
For the decoupled states $|11\rangle$ and $|00\rangle$, we define the parameter $\gamma$ as 
\begin{equation}
    \gamma = (t\overline{\omega} + (1-\kappa)t\omega_{\rm ini}),
\end{equation}
so that $S_{11} = -i\gamma$ and $S_{00} = i\gamma$. Because these are numbers, their exponentials are simply: $e^{S_{11}} = e^{-i\gamma}$ and $e^{S_{00}} = e^{i\gamma}$. For the central $2 \times 2$ block, we introduce the real parameters $x, y,$ and $z$:
\begin{align}
    z = t\overline{\omega} - (1-\kappa)t\omega_{\rm ini}, \quad
    x = (1-\kappa)t g,\quad \text{and} \quad
    y = \frac{1}{6} g(1-\kappa)t^2(\Delta\omega).
\end{align}
Substituting these in the ordered basis $\{|10\rangle, |01\rangle\}$ yields the compact form:
\begin{equation}
    S_{\rm mid} = \begin{pmatrix} -iz & -ix - y \\ -ix + y & iz \end{pmatrix}.
\end{equation}
Now, since $S_{\rm mid}^2 = -\Omega^2 \mathbb{I}_2$ with $\Omega = \sqrt{x^2 + y^2 + z^2}$, we get
\begin{equation}
    e^{S_{\rm mid}} = \cos(\Omega)\mathbb{I} + \frac{\sin(\Omega)}{\Omega}S_{\rm mid} 
    =\begin{pmatrix} 
    \cos\Omega - iz\frac{\sin\Omega}{\Omega} & (-ix - y)\frac{\sin\Omega}{\Omega} \\ 
    (-ix + y)\frac{\sin\Omega}{\Omega} & \cos\Omega + iz\frac{\sin\Omega}{\Omega} 
    \end{pmatrix}.
\end{equation}
This gives us the full $4 \times 4$ unitary evolution $U = e^S$:
\begin{equation}
    U = \begin{pmatrix}
    e^{-i\gamma} & 0 & 0 & 0 \\
    0 & \cos\Omega - iz\frac{\sin\Omega}{\Omega} & (-ix - y)\frac{\sin\Omega}{\Omega} & 0 \\
    0 & (-ix + y)\frac{\sin\Omega}{\Omega} & \cos\Omega + iz\frac{\sin\Omega}{\Omega} & 0 \\
    0 & 0 & 0 & e^{i\gamma}
    \end{pmatrix}.
\end{equation}

If the initial state of the joint system is $\rho_S(0)\otimes\tau_E$, then the final global state is given as
\begin{equation}
    \rho_{SE}(t) = U (\rho_S(0)\otimes\tau_E) U^{\dagger} = e^S (\rho_S(0)\otimes\tau_E) e^{-S}.
\end{equation}
This allows us to calculate the resulting heat and extractable work upper bound:
\begin{align}
    Q =&\ 2 \omega_{\rm ini} |\mathcal{B}|^2 (q_1 \rho_{00} - q_0 \rho_{11}) = 2 \omega_{\rm ini} |\mathcal{B}|^2 (\rho_{00} - q_0),\\
    W_{\rm ext}^{\rm upp} =&\ (\rho_{00} - \rho_{11})(\omega_{\rm fin} - \omega_{\rm ini}) + 2\omega_{\rm ini}|\mathcal{B}|^2(\rho_{00} - q_0) - \Delta E_{\rm int},\\
    \Delta E_{\rm int} =&\ - \frac{2g(\rho_{00} - q_0)\sin\Omega}{\Omega^2}\Big[y\Omega\cos\Omega + xz\sin\Omega\Big],
\end{align}
where $\rho_{ij}=\rho^{ij}_S(0)$, $q_0=(1+\exp(-\beta_{\rm ini}\omega_{\rm ini}))^{-1}$, $q_1=1-q_0$, and $|\mathcal{B}|^2 = (x^2+y^2)\left(\sin\Omega/\Omega\right)^2$.

\subsection*{The isochoric limit}
\noindent
By setting the continuous parameter $\kappa \to 0$, the adiabatic driving interval $\epsilon_1$ vanishes. The system characteristic frequency remains fixed ($\Delta\omega = 0, y = 0$), and the general unitary elegantly reduces to the pure isochoric thermalisation matrix. Applying this limiting unitary to the initial state yields the explicit reduced density matrices for the working medium and the environment at time $t$:
\begin{equation}
\rho_S(t) = \begin{pmatrix}
\rho_{11}c^2 + q_1 s^2 & e^{-i\phi}\rho_{10}\,c \\[4pt]
e^{i\phi}\rho_{01}\,c & q_1\rho_{00}c^2 + q_0(\rho_{00} + \rho_{11}s^2)
\end{pmatrix},
\quad
\rho_E(t) = \begin{pmatrix}
q_1 c^2 + \rho_{11}s^2 & ie^{-i\phi}(-1+2q_1)\rho_{10}\,s \\[4pt]
ie^{i\phi}(q_0 - q_1)\rho_{01}\,s & q_0\rho_{11}c^2 + \rho_{00}(q_0 + q_1 s^2)
\end{pmatrix}.
\end{equation}
where $c = \cos(g t)$, $s = \sin(g t)$, and $\phi = 2\omega_{\rm ini} t$. The heat exchange is the negative change in the environment's energy: 
\begin{equation}
Q = -\Delta E_E = -\text{Tr}[(\mathbb{I} \otimes H_E)(\rho_{SE}(t) - \rho_{SE}(0))] = 2\omega_{\rm ini} \sin^2(g t)(q_1\rho_{00} - q_0\rho_{11}) = 2\omega_{\rm ini} \sin^2(g t)(\rho_{00} - q_0).
\end{equation}

\section*{Efficiency of the polytropic cycle}
\noindent
The hybrid cycle consists of four sequential strokes operating between a hot bath ($r=h$) and a cold bath ($r=c$). The environment frequencies are $\omega_r$, with corresponding thermal ground-state populations $q_{r,0} = (1+e^{-\beta_r \omega_r})^{-1}$. Assuming an initial state $\rho = \text{diag}(1-p, p)$ at the start of any generic stroke, where $p$ represents the relevant population parameter, we evaluate the thermodynamic quantities using the bounds derived previously.

For a generic polytropic stroke transitioning from $\omega_i$ to $\omega_f$ in contact with bath $r$, the process is characterized by $\Delta\omega_r = \omega_f - \omega_i$ and $\overline{\omega}_r = (\omega_i + \omega_f)/2$. Defining the interaction parameters $x_r = (1-\kappa)tg$, $y_r = \frac{1}{6}g(1-\kappa)t^2\Delta\omega_r$, and $z_r = t\overline{\omega}_r - (1-\kappa)t\omega_r$, the heat and extractable work (denoted $W_{{\rm ext}}^{{\rm pol}}$) are:
\begin{align}
    Q_r^{{\rm pol}} &= 2\omega_r|\mathcal{B}_r|^2(p -q_{r,0}),\\
    W_{{\rm ext},r}^{{\rm pol}} &= (1 - 2p)\Delta\omega_r + 2\omega_r|\mathcal{B}_r|^2(p - q_{r,0}) + \frac{2g(p - q_{r,0})\sin\Omega_r}{\Omega_r^2}\Big[y_r\Omega_r\cos\Omega_r + x_rz_r\sin\Omega_r\Big],
\end{align}
where $\Omega_r=\sqrt{x_r^2+y_r^2+z_r^2}$ and $|\mathcal{B}_r|^2=(x_r^2+y_r^2)\sin^2(\Omega_r)/\Omega_r^2$.

For a generic isochoric stroke with bath $r$ at fixed frequency $\omega_r$, the work is strictly zero ($W_{{\rm ext},r}^{\rm iso}=0$), and the heat reduces to:
\begin{equation}
    Q_r^{{\rm iso}} = 2\omega_r\sin^2(g t) (p - q_{r, 0}).
\end{equation}

We map these generic equations directly to the four strokes of the steady-state limit cycle (where $\beta_r$ naturally substitutes the initial inverse temperature $\beta_{\rm ini}$ for each respective bath):
\begin{enumerate}
    \item \emph{$A\to B$ (Hot polytropic expansion):} $r=h$, $\omega_i=\omega_h \to \omega_f=\omega_c$. Initial state $\rho_A$ (population $p_A$) $\to$ output state $\rho_B$ (population $p_B$). Yields $Q_{\rm AB}$ and $W_{{\rm ext},\rm AB}^{{\rm pol}}$.
    \item \emph{$B\to C$ (Cold isochoric cooling):} $r=c$, fixed at $\omega_c$. Initial state $\rho_B$ (population $p_B$) $\to$ output state $\rho_C$ (population $p_C$). Yields $Q_{\rm BC}$ and $W_{{\rm ext},\rm BC}^{{\rm iso}} = 0$.
    \item \emph{$C\to D$ (Cold polytropic compression):} $r=c$, $\omega_i=\omega_c \to \omega_f=\omega_h$. Initial state $\rho_C$ (population $p_C$) $\to$ output state $\rho_D$ (population $p_D$). Yields $Q_{\rm CD}$ and $W_{{\rm ext},\rm CD}^{{\rm pol}}$.
    \item \emph{$D\to A$ (Hot isochoric heating):} $r=h$, fixed at $\omega_h$. Initial state $\rho_D$ (population $p_D$) $\to$ output state $\rho_A$ (population $p_A$). Yields $Q_{\rm DA}$ and $W_{{\rm ext},\rm DA}^{{\rm iso}} = 0$.
\end{enumerate}
The optimal efficiency of the cycle is the ratio of the total upper-bound work extracted during the polytropic strokes to the total heat absorbed from the hot bath:
\begin{equation}
    \mathcal{E}_{\rm pol}^{\rm opt} = \frac{W_{{\rm ext},\rm AB}^{{\rm pol}} + W_{{\rm ext},\rm CD}^{{\rm pol}}}{Q_{\rm AB} + Q_{\rm DA}}.
\end{equation}

\section*{Sign conditions for the heat currents}
\noindent
In the present hybrid cycle, the microscopic expressions for heat depend explicitly on the input state of each stroke. Hence, the standard idealised labels of \emph{hot} and \emph{cold} strokes do not automatically guarantee the desired sign pattern of heat flow. We must establish explicit conditions on the initial population of the working medium to ensure absorption during $A\to B$ and $D\to A$, and release during $B\to C$ and $C\to D$. Using the identity $q_{r,1}p_X - q_{r,0}(1-p_X) = p_X - q_{r,0}$, the exact stroke-wise heat expressions simplify to:
\begin{equation}
\begin{aligned}
    Q_{\rm AB} &= 2\omega_h\, a\, (p_A-q_{h,0}), \qquad &Q_{\rm BC} &= 2\omega_c\, s^2\, (p_B-q_{c,0}), \\
    Q_{\rm CD} &= 2\omega_c\, b\, (p_C-q_{c,0}), \qquad &Q_{\rm DA} &= 2\omega_h\, s^2\, (p_D-q_{h,0}),
\end{aligned}
\end{equation}
where $a=|\mathcal{B}_{\rm AB}|^2$, $b=|\mathcal{B}_{\rm CD}|^2$, and $s^2=\sin^2(gt)$. Assuming all interaction parameters are strictly positive ($0<a<1$, $0<s^2<1$, $b>0$), the desired thermodynamic operation ($Q_{\rm AB}>0$, $Q_{\rm BC}<0$, $Q_{\rm CD}<0$, $Q_{\rm DA}>0$) is equivalent to satisfying the population bounds:
\begin{equation}
    p_A>q_{h,0},\qquad
    p_B<q_{c,0},\qquad
    p_C<q_{c,0},\qquad
    p_D>q_{h,0}.
\label{eq:desired-sign-conditions}
\end{equation}

\noindent
\begin{proposition}
Assume $q_{h,0} < q_{c,0}$, and choose a diagonal initial state $\rho_A = \operatorname{diag}(1-p_A, p_A)$ with the initial ground-state population restricted to the interval:
\begin{equation}
    q_{h,0} < p_A < \frac{q_{c,0}-a q_{h,0}}{1-a}.
\label{eq:initial-interval}
\end{equation}
Under these conditions, the cycle strictly satisfies the desired heat-flow directions.    
\end{proposition}

\begin{proof}
We verify the bounds sequentially throughout the cycle:
\begin{enumerate}
    \item \textbf{Stroke $A \to B$:} By the lower bound in Eq.~\eqref{eq:initial-interval}, $p_A > q_{h,0}$, yielding $Q_{\rm AB} > 0$.
    
    \item \textbf{Stroke $B \to C$:} From the polytropic map of Stroke 1, the ground-state population evolves as
\begin{equation}
    p_B=(1-a)p_A+a q_{h,0}.
\label{eq:pB-map}
\end{equation}
Using the upper bound in Eq.~\eqref{eq:initial-interval},
\begin{align*}
    p_B &= (1-a)p_A+a q_{h,0} 
    < (1-a)\frac{q_{c,0}-a q_{h,0}}{1-a}+a q_{h,0} = q_{c,0}.
\end{align*}
Since $Q_{\rm BC}=2\omega_c s^2(p_B-q_{c,0})$ and $s^2>0$, we obtain $Q_{\rm BC}<0$. Moreover, $p_B - q_{h,0} = (1-a)(p_A - q_{h,0}) > 0$, establishing that $p_B > q_{h,0}$ for subsequent steps.

    \item \textbf{Stroke $C \to D$:} The map yields $p_C = (1-s^2)p_B + s^2 q_{c,0}$. Subtracting $q_{c,0}$ gives $p_C - q_{c,0} = (1-s^2)(p_B - q_{c,0})$. Since $p_B < q_{c,0}$ from Step 2, we have $p_C < q_{c,0}$, guaranteeing $Q_{\rm CD} < 0$. Additionally, subtracting $q_{h,0}$ yields:
    \begin{equation*}
        p_C - q_{h,0} = (1-s^2)(p_B - q_{h,0}) + s^2(q_{c,0} - q_{h,0}).
    \end{equation*}
    Since $p_B > q_{h,0}$ and we assumed $q_{c,0} > q_{h,0}$, both terms are positive, meaning $p_C > q_{h,0}$.
    
    \item \textbf{Stroke $D \to A$:} Finally, $p_D = (1-b)p_C + b q_{c,0}$. Subtracting $q_{h,0}$ gives:
    \begin{equation*}
        p_D - q_{h,0} = (1-b)(p_C - q_{h,0}) + b(q_{c,0} - q_{h,0}).
    \end{equation*}
    Using $p_C > q_{h,0}$ from Step 3 and $q_{c,0} > q_{h,0}$, the right-hand side is strictly positive. Thus, $p_D > q_{h,0}$, which yields $Q_{\rm DA} > 0$.
\end{enumerate}
\end{proof}

\begin{figure*}
    \captionsetup[subfigure]{labelformat=empty}
    \centering
    \subfloat[(a)]{\includegraphics[width=0.49\textwidth]{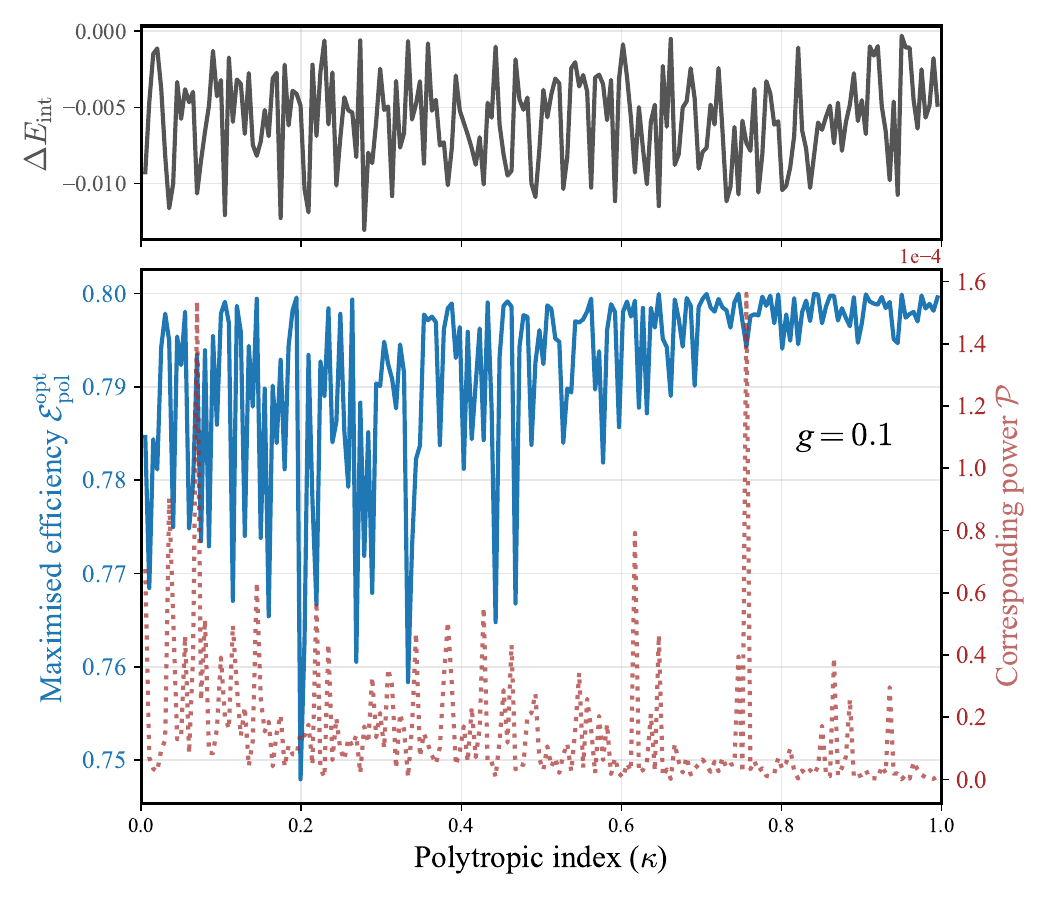}\label{fig:max_eff2}}\hfill
    \subfloat[(b)]{\includegraphics[width=0.49\textwidth]{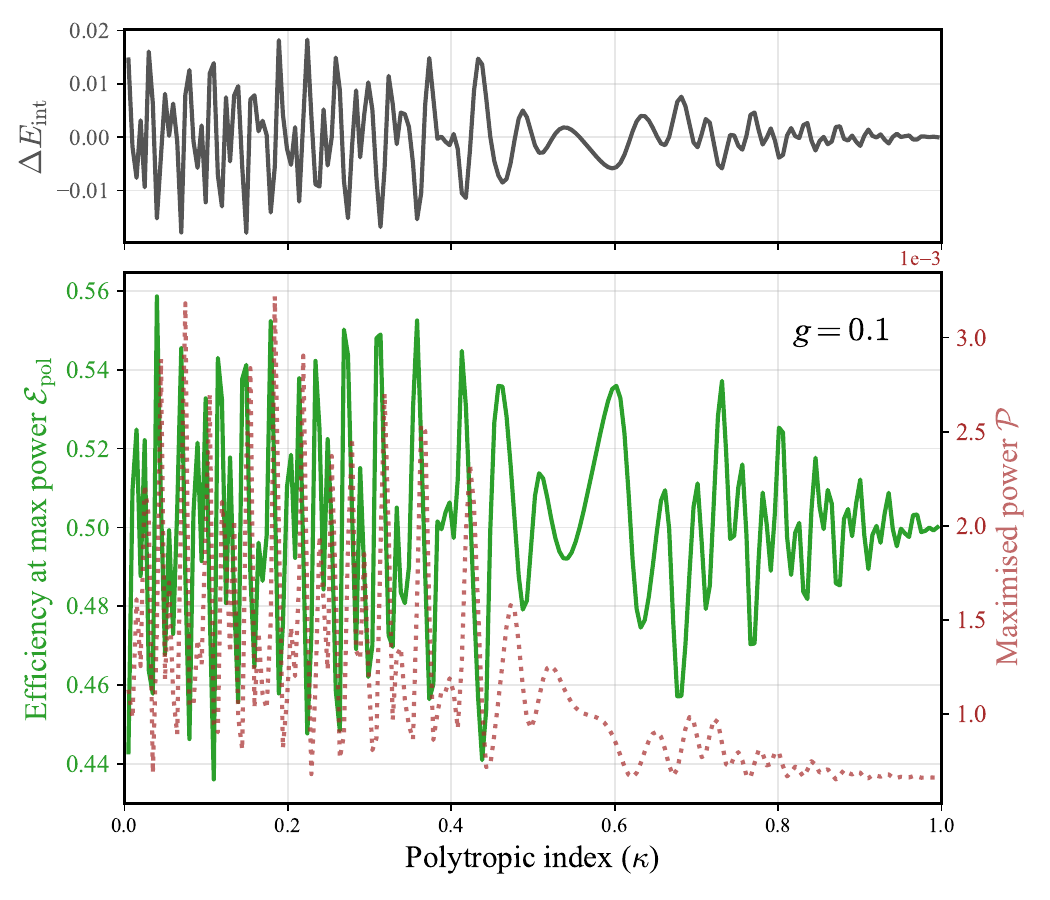}\label{fig:max_pow2}}\\
    \subfloat[(c)]{\includegraphics[width=0.49\textwidth]{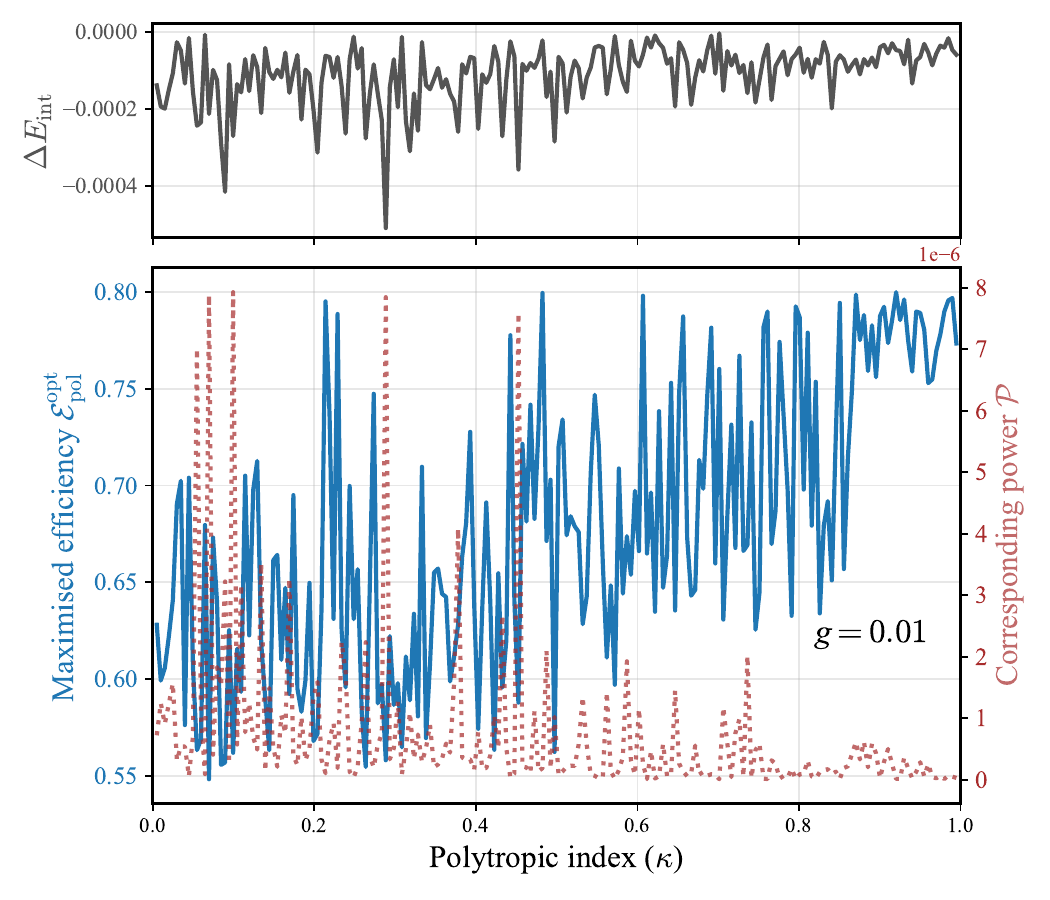}\label{fig:max_eff2}}\hfill
    \subfloat[(d)]{\includegraphics[width=0.49\textwidth]{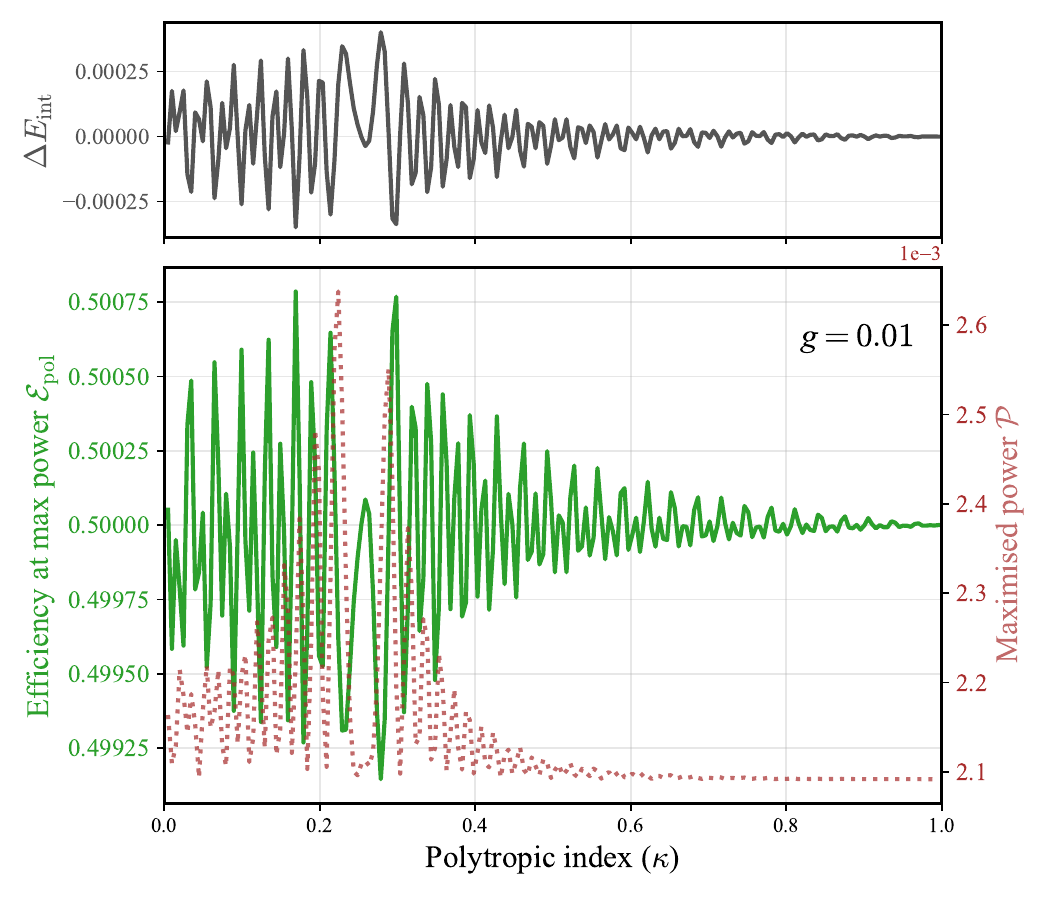}\label{fig:max_pow2}}
    \caption{\textbf{Thermodynamic performance and the finite-bath interaction cost of the hybrid cavity-optomechanical cycle.} Everything remains the same as in Fig.~\ref{fig:performance}, except for $g=0.1$ (stronger system-bath coupling) in the top row [(a) and (b)], and  $g=0.01$ (weaker system-bath coupling) in the bottom row [(c) and (d)].\label{fig:performancevariance}} 
\end{figure*}

\section{Computational methodology}\label{sec:method}
\noindent
We simulated the engine's thermodynamic cycle by tracking the joint unitary evolution of the system and bath across four discrete strokes. To keep the computational framework rigorous, we split the architecture into two modes. The first was a single-shot mode, used to validate our discrete-time matrices against analytical models by forcing a predefined initial state. The second was a steady-state mode designed to evaluate a generalised stable cycle. To generate performance data for this final mode without computing expensive transient cycles, we mapped the dynamics into Liouville space. By representing the entire four-stroke sequence as a single linear superoperator, we algebraically extracted the steady-state density matrix ($\rho_{\rm steady}$) directly from the unit eigenvalue of the superoperator. 

From $\rho_{\rm steady}$, we calculated net work, heat flow, and non-Markovian interaction energy penalties ($\Delta E_{\text{int}}$) using relative and von Neumann entropies. Finally, to locate the highest efficiency ($\mathcal{E}_{\text{pol}}^{\text{opt}}$) and corresponding power ($\mathcal{P}$) across the polytropic index $\kappa$, we deployed a $\kappa$-dependent grid search over the interaction time $t$. We dynamically scaled the upper bound of this search space using a cubic function ($t_{\max} \sim \kappa^3$). This non-linear scaling exploits a natural physical dampening in the system: as the cycle approaches the adiabatic limit ($\kappa \to 1$), the first-order Trotter errors (arising from the commutator $[H_{\rm iso}, H_{\rm adi}]$) are suppressed. While these off-diagonal error terms scale as $\mathcal{O}(t^2)$, the cubic proxy safely navigates this constraint. It tightly restricts the interaction time at low $\kappa$ to prevent rapid phase scrambling, yet expands the search space aggressively at high $\kappa$ to capture the increasingly sparse quasi-static resonances. By cleanly cutting off the search before numerical dephasing can dominate, the dynamic grid preserves the fragile non-Markovian coherent resonances and prevents the system from artificially collapsing to the memoryless Markovian limit.

\end{document}